\documentclass[11pt]{article}

\usepackage{graphicx, amssymb, latexsym, amsfonts, amsmath, lscape, amscd,
amsthm, color, epsfig, mathrsfs, tikz, enumerate}
\usepackage{subcaption}
\usepackage{csquotes}
\usepackage{array, multirow, boldline}
\usepackage{soul}
\usepackage{makecell}
\setcellgapes{5pt}

\usepackage{amssymb}
\usepackage{color}
\usepackage{mathtools}
\usepackage{pgf,tikz}
\usetikzlibrary{decorations.pathreplacing}
\usetikzlibrary{decorations.markings}
\usetikzlibrary{positioning}
\usetikzlibrary{arrows}

\pgfdeclarelayer{bg}    
\pgfsetlayers{bg,main}  

\usepackage[pagewise]{lineno}

\newtheorem{theorem}{Theorem}[section]

\newtheorem{lemma}[theorem]{Lemma}
\newtheorem{proposition}[theorem]{Proposition}

\newtheorem{problem}[theorem]{Problem}
\newtheorem{observation}[theorem]{Observation}

\theoremstyle{definition}
\newtheorem{remark}[theorem]{Remark}

\newcommand\DELETE[1]{}

\begin{document}


\title{{\bf On homomorphism related parameters of oriented triangle-free planar graphs}}
\author{
{\sc Soura Sena Das}$\,^{a}$, {\sc Soumen Nandi}$\,^{b}$, {\sc Sagnik Sen}$\,^{c}$ \\
\mbox{}\\
{\small $(a)$ Indian Statistical Institute, Kolkata, India}\\
{\small $(b)$ Netaji Subhas Open University, India}\\
{\small $(c)$ Indian Institute of Technology Dharwad, India}
}

\date{}

\maketitle

\begin{abstract}
The first major contribution of this work is proving that the 
oriented relative clique number of oriented triangle-free planar graphs is $10$, which completely answers and closes an open problem posed by Sopena    (Discrete Mathematics 2016) in the most recent survey on oriented colorings. The second major contribution of the paper is to prove that if all oriented triangle-free planar graphs admit a homomorphism to a particular oriented graph $\overrightarrow{T}$, then its underlying graph $T$ must have minimum degree  at least $10$. 
This result implies that, for the family of 
oriented  triangle-free planar graphs,  
the lower bounds of the parameters 
oriented chromatic number, 
pushable chromatic number, 
$2$-dipath $L(p,1)$-labeling span, 
and oriented $L(p,1)$-labeling span 
are at least $11$, $6$, $p+8$, and $2p+8$, respectively, 
where $p \geq 1$. 
That is, we are able to obtain improved lower bounds of a number of other parameters restricted to the family of oriented triangle-free planar graphs using our second major contribution. 

\medskip

\end{abstract}

\section{Introduction}
An \textit{oriented graph} $\overrightarrow{G}$ is a directed graph without loops or parallel arcs in opposite directions.
Notice that, the definition allows parallel arcs in the same direction between two vertices, however, for our purpose, the extra parallel arcs in the same direction are redundant. 
Therefore, we can assume that all oriented graphs mentioned in this article have a simple graph as it's underlying graph. 
Moreover, given an oriented graph $\overrightarrow{G}$, we will denote its set of vertices and arcs by $V(\overrightarrow{G})$ and $A(\overrightarrow{G})$, respectively. 
We denote its underlying graph as $G$, that is, simply by dropping the overhead arrow. 

In 1992, Courcelle~\cite{courcelle-monadic} introduced the concept of oriented coloring and chromatic number in one of his papers from the series which established the illustrious Courcelle's theorem~\cite{COURCELLE199012}. 
A \textit{homomorphism }of an oriented graph $\overrightarrow{G}$ to another oriented graph $\overrightarrow{H}$ is a vertex mapping 
$f: V(\overrightarrow{G}) \to V(\overrightarrow{H})$
such that for any arc 
$xy \in A(\overrightarrow{G})$, there is an arc from $f(x)$ to $f(y)$ in $\overrightarrow{H}$. 
In particular, if $\overrightarrow{G}$ admits a homomorphism to 
$\overrightarrow{H}$ and $\overrightarrow{H}$ has $k$ vertices, then we say that 
$\overrightarrow{G}$ admits an oriented \textit{$k$-coloring}.
The \textit{oriented chromatic number} of $\overrightarrow{G}$, denoted by $\chi_o(\overrightarrow{G})$, is the minimum $k$  for which $G$ admits an oriented $k$-coloring.

\begin{remark}\label{rem 1}
It is possible to define the ordinary chromatic number of simple graphs in a similar fashion as pointed out by Hell and Ne\v{s}et\v{r}il~\cite{hell}. From this point of view, it is easier to observe how the oriented chromatic number is a natural analogue of the ordinary chromatic number.  For the notion of clique number, its oriented analogue ramifies into two different parameters, namely, oriented relative and absolute clique numbers. 
\end{remark}

An \textit{oriented relative clique} 
$R \subseteq V(\overrightarrow{G})$ of $\overrightarrow{G}$ is a vertex subset 
 for which any two vertices $x,y \in R$
 must satisfy $f(x) \neq f(y)$ under any homomorphism of $\overrightarrow{G}$ to any $\overrightarrow{H}$. 
  The \textit{oriented relative clique number} of 
$\overrightarrow{G}$, denoted by $\omega_{ro}(\overrightarrow{G})$, 
is the maximum $k$ such that $\overrightarrow{G}$ contains an oriented relative clique $R$ having $|R|=k$. 
An \textit{oriented absolute clique} 
$A \subseteq V(\overrightarrow{G})$ of $\overrightarrow{G}$ is a vertex subset 
 for which any two vertices $x,y \in A$
 must satisfy $f(x) \neq f(y)$ under any homomorphism of $\overrightarrow{G}[A]$ to any $\overrightarrow{H}$, where $\overrightarrow{G}[A]$ is the subgraph of $\overrightarrow{G}$ induced by the vertex subset $A$. 
  The \textit{oriented absolute clique number} of 
$\overrightarrow{G}$, denoted by $\omega_{ao}(\overrightarrow{G})$, 
is the maximum $k$ such that $\overrightarrow{G}$ contains an oriented absolute clique $A$ having $|A|=k$.

As both the notions use the condition of two vertices having distinct images under any homomorphism, the following observation due to Klostermeyer and MacGillivray~\cite{36} becomes very useful in their study.

\begin{proposition}[Klostermeyer and MacGillivray 2004~\cite{36}]\label{obs distinct image if and only if}
Let $\overrightarrow{G}$ be an oriented graph. 
Two vertices $x,y$ of $\overrightarrow{G}$ cannot have the same image under any homomorphism of $\overrightarrow{G}$ to any oriented graph $\overrightarrow{H}$ if and only if either they are adjacent or they are connected by a directed $2$-path (a directed path with $2$ arcs). 
\end{proposition}

\begin{remark}
The main difference  between an oriented relative clique $R$ and an oriented absolute clique $A$ is 
the following: 
the non-adjacent vertices of $R$ can be connected by a directed $2$-path whose internal vertex does not belong to $R$, while the non-adjacent vertices of $A$ must be connected by a directed $2$-path whose internal  
vertex belongs to $A$. 
\end{remark}

From the definitions of the three parameters, a natural relation between them becomes apparent. This relation makes a study of one of these parameters impact the study of the other two.

\begin{proposition}[Nandy, Sen, and Sopena 2016~\cite{NSS}]\label{prop relation between the three parameters}
For any oriented graph $\overrightarrow{G}$ we have
$$\omega_{ao}(\overrightarrow{G}) \leq \omega_{ro}(\overrightarrow{G}) \leq \chi_{o}(\overrightarrow{G}).$$
\end{proposition}

The definitions of the three parameters are extended to simple graphs and graph families for the ease of presentation of the  research results. 
Let $p \in \{\chi_o, \omega_{ro}, \omega_{ao}\}$ denote any of the three parameters. Then, for a simple graph $G$ we have 
$$p(G) = \max\{p(\overrightarrow{G}): \overrightarrow{G} \text{ is an orientation of } G\}$$  
and for a  family $\mathcal{F}$ of  graphs we have  $$p(\mathcal{F})=\max\{p(G) : G \in \mathcal{F}\}.$$
Given a family $\mathcal{F}$ of oriented graphs (resp., graphs), 
if $\forall\overrightarrow{G} \in \mathcal{F}$ 
(resp., all orientations $\overrightarrow{G}$ of $G$ for $G \in \mathcal{F}$) admits a homomorphism to $\overrightarrow{H}$, 
then $\overrightarrow{H}$ is 
a bound of $\mathcal{F}$. 
Additionally, if no 
proper subgraph of  $\overrightarrow{H}$ is a
bound of $\mathcal{F}$, then we say that $\overrightarrow{H}$ is a 
\textit{minimal bound} of $\mathcal{F}$.

\begin{remark}\label{remark oriented chromatic minimal family bounds}
    If $\chi_o(\mathcal{F}) = k$, then it is not necessary that a (minimal) bound of 
    $\mathcal{F}$ will exist on $k$ vertices. For example, for the family $\mathcal{T}_3$  of all tournaments on $3$ vertices, we have $\chi_o(\mathcal{T}_3) = 3$, but there exists no (minimal) bound of $\mathcal{T}_3$ on $3$ vertices. On the other hand, Sopena (see Proposition~2 in~\cite{SOPENA2002309}) showed that for the family $\mathcal{P}_3$ of planar graphs, there must exist a (minimal) bound on $\chi_o(\mathcal{P}_3)$ vertices. Similarly, one can show that 
    for the family $\mathcal{P}_g$ of planar graphs having girth at least $g$, there must exist a (minimal) bound on $\chi_o(\mathcal{P}_g)$ vertices.  Thus, proving the non-existence of a (minimal) bound
    for such families on certain numbers of vertices helps prove lower bounds for the oriented chromatic number of the said family. 
\end{remark}

\subsection{Context and motivation}
The initial problem considered in Courcelle's work~\cite{courcelle-monadic} was to find the oriented chromatic number of the family of planar graphs. 
 Courcelle~\cite{courcelle-monadic} provided an upper bound for oriented chromatic number   as $64 \cdot 3^{63}$ for oriented planar graphs with in-degree at most $3$, which was soon improved to $80$ for general oriented planar graphs by Raspaud and Sopena~\cite{planar80}.
To date, according to 
the best known~\cite{marshall18,planar80} lower and upper bounds of the oriented chromatic number of planar graphs, its exact value is known to lie somewhere between 
$18$ and $80$. The resolution of the problem, that is, finding the exact value of the oriented chromatic number of planar graphs, will provide us with the analogue of the Four-Color Theorem for oriented graphs, while on the other hand, the problem seems difficult to solve and has remained open since $1994$.

As related sub-problems, the oriented chromatic number of the family, $\mathcal{P}_g$, of planar graphs having girth at least $g$ has been given a particular focus~\cite{Borodin2005,Borodin2007,mad, marshall18, marshallgirth6,NRSnotreza,Ochemgirth4,OCHEM200882,OPgirth4,planar80}. 
Naturally, the oriented relative and absolute clique numbers for the same families are studied as well~\cite{mj15,NSS}. 
To give an overview of the state of the art, let us present 
the related results in Table~\ref{table results}.

Observe that, the problem of finding the exact values of 
$\omega_{ao}(\mathcal{P}_g)$ for all $g \geq 3$ is completely solved. 
It is worth mentioning that Klostermeyer and MacGillivray~\cite{36}  showed  
$15 \leq \omega_{ao}(\mathcal{P}_3) \leq 36$ and 
$\omega_{ao}(\mathcal{P}_4) \leq 14$ in $2004$. 
Furthermore, in the same work~\cite{36} the authors conjectured $\omega_{ao}(\mathcal{P}_3) = 15$ and 
left the determination of the exact value of $\omega_{ao}(\mathcal{P}_4)$ as an open problem. Later, in~\cite{NSS}, the conjecture was positively settled by proving 
$\omega_{ao}(\mathcal{P}_3) = 15$ and the open problem was closed by proving $\omega_{ao}(\mathcal{P}_4) = 6$.

In~\cite{mj15}, the oriented relative clique number for the family, $\mathcal{P}_g$, of planar graphs having girth at least $g$ was studied for all $g \geq 3$.  
In that paper~\cite{mj15}, the exact values of 
$\omega_{ro}(\mathcal{P}_g)$ is
determined for all $g \geq 5$,  
while the main results of the same were 
given by $15 \leq \omega_{ro}(\mathcal{P}_3) \leq 32$
and $10 \leq \omega_{ao}(\mathcal{P}_4) \leq 14$. 
Thus, finding the exact values of 
$\omega_{ro}(\mathcal{P}_3)$ and 
$\omega_{ro}(\mathcal{P}_4)$ are natural open problems. In fact, these two were later 
highlighted in the latest survey on oriented coloring by 
Sopena~\cite{sopena-updated-survey} as major open problems in the domain (see Problem~3
in Section 10: ``Open Problems''). 

\begin{problem}[Sopena 2016~\cite{sopena-updated-survey}]\label{problem planar}
Determine the largest possible value of the 
oriented relative clique number of  planar graphs.                 
\end{problem}

\begin{problem}[Sopena 2016~\cite{sopena-updated-survey}]\label{problem triangle-free}
Determine the largest possible value of the 
oriented relative clique number of triangle-free planar graphs.
\end{problem}

\begin{table}[t]
\centering 
\makegapedcells
\begin{tabular}{|c|l|l|l|}
\hline
\hline  
    $\mathcal{P}_g$, for $g$  & \multicolumn{1}{|c|}{$\chi_o(\mathcal{P}_g)$} &
     \multicolumn{1}{|c|}{$\omega_{ro}(\mathcal{P}_g)$} &  \multicolumn{1}{|c|}{$\omega_{ao}(\mathcal{P}_g)$}  \\
    \hline
    \hline
    $=3$ & $18 \leq \chi_o(\mathcal{P}_3) \leq 80$~\cite{marshall18,planar80} & {$15 \leq \omega_{ro}(\mathcal{P}_3) \leq 32$}~\cite{mj15}   & $\omega_{ao}(\mathcal{P}_3) = 15$~\cite{NSS} \\
    \hline
    $=4$ & $\textcolor{red}{\textbf{11} \leq \mathbf{\chi_o(\mathcal{P}_4)}} \leq 40$~\cite{Ochemgirth4,OPgirth4} & $\textcolor{red}{\mathbf{\omega_{ro}(\mathcal{P}_4)} = \textbf{10}}$ &   $\omega_{ao}(\mathcal{P}_4) = 6$~\cite{NSS}   \\
    & \textbf{[Theorem~\ref{thm orientedChromatic11}]} & \textbf{[Theorem~\ref{th main}]} & \\
    \hline
    $=5$ & $7 \leq \chi_o(\mathcal{P}_5) \leq 16$~\cite{marshallgirth6,Pinlougirth5} & $\omega_{ro}(\mathcal{P}_5) = 6$~\cite{mj15}   & $\omega_{ao}(\mathcal{P}_5) = 5$~\cite{NSS}  \\
    \hline
    $=6$ & $ 7 \leq \chi_o(\mathcal{P}_6) \leq 11 $~\cite{marshallgirth6,mad} & $\omega_{ro}(\mathcal{P}_6) = 4$~\cite{mj15}  &  \multirow{7}{*}{$\omega_{ao}(\mathcal{P}_{g \geq 6}) = 3$~\cite{NSS}} \\
    \cline{1-3}
    $=7 $ & $ 6 \leq \chi_o(\mathcal{P}_7) \leq 7$~\cite{NRSnotreza,Borodin2005} & \multirow{6}{*}{$\omega_{ro}(\mathcal{P}_{g \geq 7}) = 3$~\cite{mj15}} &  \\
    \cline{1-2}
    $=8,9,10 $ & $ 5 \leq \chi_o(\mathcal{P}_{g}) \leq 7 $~\cite{NRSnotreza,Borodin2005} &  &  \\
    \cline{1-2}
    $=11 $ & $ 5 \leq \chi_o(\mathcal{P}_{11}) \leq 6$~\cite{NRSnotreza,OCHEM200882} &  &  \\
    \cline{1-2}
    $ \geq 12 $ & $ \chi_o(\mathcal{P}_{g \geq 12}) = 5 $~\cite{NRSnotreza,Borodin2007} &  &  \\
    \hline
\end{tabular}

\caption{This is the list of all known lower and upper bounds for $\chi_o(\mathcal{P}_g) $,
    $\omega_{ro}(\mathcal{P}_g)$ and $\omega_{ao}(\mathcal{P}_g)$ where $\mathcal{P}_g$ denotes the family of planar graphs having girth at least $g$.}
\label{table results}
\end{table}

\medskip

\noindent \textit{The significance of oriented cliques:}
The notions of oriented relative and absolute clique number are together referred to as oriented cliques. Even though there are only a handful of articles dedicated to the structural~\cite{mj15,36,NSS} 
and  algorithmic study~\cite{ COELHO2023, dybizbanski_coloring_2022} of oriented cliques, the notion of oriented cliques is natural and important, and they appear (sometimes implicitly) in several research papers on oriented coloring and related topics, namely, papers related to lower bounds for oriented chromatic number~\cite{mj15,marshall17,marshall18,marshallgirth6,Ochemgirth4,Ochem_negativeresults,orientedchi}, oriented achromatic number~\cite{SOPENA2014102,PavanSopenaAchromatic}, deeply critical oriented graphs~\cite{borodin2001deeply,pavan} and the oriented analogue of the degree diameter problem~\cite{erdos1966problem,furedi1998minimal,katona1967problem,kostochka1999minimum}.
Overall, it will not be unfair to claim that the structural understanding of the oriented cliques is a fundamental problem and is worth studying.

\medskip

The second column of  Table~\ref{table results} shows that the exact values of $\chi_{o}(\mathcal{P}_g)$ remains unknown for all $g \in \{3, 4, \ldots, 11\}$. By Remark~\ref{remark oriented chromatic minimal family bounds}, finding the exact values of $\chi_{o}(\mathcal{P}_g)$ boils down to finding (minimal) bounds of $\mathcal{P}_g$ on $\chi_o(\mathcal{P}_g)$ vertices.
In fact, most of the lower and upper bounds
in the second column of Table~\ref{table results} are obtained theoretically by inspecting (explicitly or implicitly)  the properties 
of the minimal bounds of 
$\mathcal{P}_g$. One of the major exceptions is the case of 
 $\chi_{o}(\mathcal{P}_4)$, where the lower bound of $11$ 
 is achieved by Ochem~\cite{Ochemgirth4} through a computer check, and hence the result lacks a theoretical proof, implying a natural open problem. 

\begin{problem}\label{prob theoretical}
    Provide a theoretical proof of $\chi_{o}(\mathcal{P}_4) \geq 11$. 
\end{problem}

\subsection{Our contributions and organisation}
In this article, we present our section-wise contributions as organized below.

\begin{itemize}
    \item In \textbf{Section~\ref{sec prelim}}, we present the definitions, notation, and terminology specific to oriented graphs. 

    \item In \textbf{Section~\ref{sec proof}}, we prove that the oriented relative clique number for the family of triangle-free planar graphs is $10$, and hence solve Problem~\ref{problem triangle-free} due to Sopena~\cite{sopena-updated-survey}. The proof is novel and lengthy. 

    \item In \textbf{Section~\ref{sec min_bounds}}, we use oriented cliques to show that any minimal bound of the family of triangle-free planar graphs have minimum degree at least $10$. 
    As a consequence of this result, we obtain the first theoretical proof of the lower bound $\chi_o(\mathcal{P}_4) \geq 11$, and thus solve Problem~\ref{prob theoretical}. 
    Moreover, we apply the main result of this section to establish improved lower bounds in the context of three other parameters related to oriented graphs, namely, pushable chromatic number, $2$-dipath $L(p,1)$-span, and oriented $L(p,1)$-span (definitions provided in Section~\ref{sec min_bounds}).

    \item In \textbf{Section~\ref{sec conclusion}}, we conclude this article and present a few possible future directions of study.
\end{itemize}

\medskip

\noindent \textbf{Note:} A preliminary version of this article was presented as a paper in CALDAM 2020~\cite{DasN020}. 
This version contains a significant improvement of the proof of Theorem~\ref{th main} and its presentation along with newer results in Section~\ref{sec min_bounds}.

\section{Preliminaries}~\label{sec prelim}
We  follow West~\cite{D.B.West} for standard graph theory notation.
In general, if some particular notation is meaningful for an undirected graph, but not for an oriented graph, then the same, if applied in the context of an oriented graph
$\overrightarrow{G}$, will
be actually meant for its underlying graph $G$. 
Further, some useful, but non-standard notations are presented below.

Let $\overrightarrow{G}$ be an oriented graph, and let $xy$ be an arc of it. 
Then we say that $x$ is an \textit{in-neighbor} of $y$ and $y$ is an \textit{out-neighbor} of $x$. 
 The set of all in-neighbors and out-neighbors of $x$ is denoted by $N^-(x)$
 and $N^+(x)$, respectively. 
 Moreover,  \textit{in-degree} and \textit{out-degree} of a vertex $x$ 
 is given by  $d^-(x) = |N^-(x)|$ and 
 $d^+(x) = |N^+(x)|$, respectively. 
 Two vertices $x,y$ \textit{agree} on a third vertex $z$ if 
 $z \in N^{\alpha}(u) \cap N^{\alpha}(v)$ for some $\alpha \in \{+,-\}$.  
 Also $x,y$ \textit{disagree} on $z$ if  
 $z \in N^{\alpha}(x) \cap N^{\beta}(y)$ for some 
 $\{\alpha, \beta\} = \{+,-\}$\footnote{We use this notation frequently to denote $\alpha, \beta \in \{+,-\}$ and $\alpha \neq \beta$. Our notation is a set theoretic equation whose solutions are the values that $\alpha, \beta$ may take.}.
 We say a vertex $x$ \textit{sees} a vertex $y$ if they are adjacent, or they are connected by a directed $2$-path. If $x,y$ are connected by a directed $2$-path with the internal vertex of the directed $2$-path being $z$, we say that $x$ sees $y$ \textit{via} $z$.

 A cycle $C = x_1x_2 \cdots x_k x_1$ is a \textit{separating cycle} of a connected graph $G$ if $G - C$, that is, the
graph obtained by deleting the vertices of the cycle $C$ from $G$, is disconnected. Moreover,
if two vertices $y$ and $z$ belong to different components of $G - C$, then we say that $C$
\textit{separates} $y$ and $z$.

\section{Oriented relative clique number of $\mathcal{P}_4$}\label{sec proof}
In this section we are going to respond to Problem~\ref{problem triangle-free} due to 
Sopena~\cite{sopena-updated-survey} by giving a conclusive answer via proving the following theorem. 

\begin{theorem}\label{th main}
For the family $\mathcal{P}_4$ of triangle-free planar graphs, we have
$\omega_{ro}(\mathcal{P}_4)=10$.
\end{theorem}

The rest of the section is dedicated to the proof of the main result. As the proof is novel and lengthy, and thus we have separated it into relevant subsections for the convenience of the readers. Moreover, there are several figures to improve the readability of the proof. 

\subsection{Basic set-up of the proof}
Let $\overrightarrow{H}$ be a minimal
counter-example to Theorem~\ref{th main}
with respect to  $|V(\overrightarrow{H})|+|A(\overrightarrow{H})|$.  
That means, $\overrightarrow{H}$ is an oriented triangle-free planar  graph
and its oriented relative clique number is strictly greater than $10$. Moreover, any other oriented  triangle-free planar graph $\overrightarrow{H'}$ having oriented relative clique number strictly greater than $10$ must have 
$|V(\overrightarrow{H}')|+|A(\overrightarrow{H}')| \geq |V(\overrightarrow{H})|+|A(\overrightarrow{H})|$.

Let us assume that $R$ is an oriented relative clique of $\overrightarrow{H}$ 
having $\omega_{ro}(\overrightarrow{H})$ many vertices. 
In particular, $|R| \geq 11$. 
For convenience, we will call the vertices of $R$ as \textit{good} vertices. 
On the other hand, the vertices that are not part of $R$, that is, 
the vertices from the set $S= V(\overrightarrow{H}) \setminus R$ are
 \textit{helper} vertices. 
For the rest of this proof, we are going to assume a particular 
planar embedding of $\overrightarrow{H}$ and whenever we speak about a subgraph of $\overrightarrow{H}$, we will assume it to have the induced embedding. 
Now, we are going to make two observations about the helper vertices which directly follow from the minimality of $\overrightarrow{H}$. 

\begin{observation}
    The set $S$ of helper vertices is independent. 
\end{observation}

\begin{proof}
   If two helper vertices $h_1$ and $h_2$ are adjacent with the arc $h_1h_2$, then consider the oriented graph $\overrightarrow{H} - h_1h_2$ obtained by deleting the arc $h_1h_2$ from $\overrightarrow{H}$. Notice that, the vertices of $R$ still see each other in $\overrightarrow{H} - h_1h_2$. Thus, $\overrightarrow{H} - h_1h_2$ is a counter-example to Theorem~\ref{th main} which 
   contradicts the minimality of $\overrightarrow{H}$. 
\end{proof}

\begin{observation}
    If $h$ is a helper vertex, then it has at least one in-neighbor and at least one out-neighbor. In particular, its degree is at least two, that is, $d(h) \geq 2$. 
\end{observation}

\begin{proof}
   Let $h$ be a helper vertex. If all its neighbors are in-neighbors (resp., out-neighbors), then none of the good vertices are seeing each other via $h$. That implies, $R$ is a relative clique in $\overrightarrow{H} - h$. 
   Thus, $\overrightarrow{H} - h$ is a counter-example to Theorem~\ref{th main} which contradicts the minimality of $\overrightarrow{H}$. 
\end{proof}

 This enables us to draw some information about the neighborhood of the helper vertices. Next we will recall a result due to Klostermeyer and MacGillivray~\cite{36} 
 which provides us some information about the number of good vertices in the neighborhood of a vertex. Notice that, this result was proved by 
 Klostermeyer and MacGillivray~\cite{36} (implicitly proved as a part of the proof of Theorem 11 in~\cite{36}, for a complete proof see~\cite{Chakraborty0NR023}). 

 \begin{lemma}[Klostermeyer and MacGillivray 2004~\cite{36}]\label{lem 5}
     Let $v$ be any vertex of $\overrightarrow{H}$. At most four good vertices may agree on $v$.  
 \end{lemma}
 
 Continuing in the spirit of the above lemma, 
 we are going to show that it is not possible for $\overrightarrow{H}$ to contain certain configurations. 
 These forbidden configurations hold the key for the proof.

\subsection{Special terms used in the proof} 
Given a planar graph and its embedding, any cycle $C$ is drawn as a simple close curve. By Jordan Curve Theorem, the curve induced by $C$ divides the plane into two connected (topologically) regions. 
Usually, in some of the proofs, 
 we will work with a subgraph of $\overrightarrow{H}$. Moreover, the embedding of this subgraph will be is unique up to continuous deformation, and all the 
   assumptions related to its orientations are taken without loss of generality. Therefore, it is possible for us to 
   draw a generic picture of this subgraph and 
   name the connected regions induced by it in a figure with references. Moreover, if $F$ denotes such a region, then by ``$F$ is empty'' we mean that it does not contain a good vertex, and  $g \in F$ will mean that $g$ is a good vertex that belongs to $F$. Also, $|F|$ will denote the number of good vertices belonging in it.

 \subsection{Forbidding four good vertices in a neighborhood}\label{subsec forbid4}
 In this subsection, we are going to show that
 any vertex $x$ of $\overrightarrow{H}$ 
 can have at most three good vertices in its neighborhood. To prove so, we will show that if any vertex $v$ of $\overrightarrow{H}$ contains four good vertices in its neighborhood, then we arrive to a contradiction, and thus forbid such a configuration in $\overrightarrow{H}$.  
 However, based on the embedding and the orientation of the configuration, 
 we have to consider a few different cases in the lemmas listed below.

 \begin{lemma}\label{lem 4,0}
 Four good vertices $v_1,v_2,v_3, v_4$ cannot agree with each other
  on a vertex $v$.
 \end{lemma}
  
  \begin{proof}
  Without loss of generality, assume that $v_1, v_2, v_3, v_4$ are arranged in a clockwise order around $v$ in the planar embedding of $\overrightarrow{H}$ and that they are in-neighbors of $v$. 
    Note that $v_2$ must see $v_4$ via some $h_1$ as $\overrightarrow{H}$ is triangle-free. Observe that the cycle $vv_2h_1v_4v$  separates $v_1$, and $v_3$.  
    Therefore, $v_1$ is forced to see $v_3$ via $h_1$ as they agree on $v$, and if either of them are adjacent to $v_2$ or $v_4$, then it will result in a triangle. 
    Without loss of generality, 
    we may assume that  
    $v_1, v_2 \in N^-(h_1)$, and 
    that  $v_3, v_4 \in N^+(h_1)$. 
   Furthermore, $v_1$ must see $v_2$ via some $h_2$, and $v_3$ must see $v_4$ via some $h_3$. As $h_2$ cannot see $v_4$, it is a helper. Similarly,  as $h_3$ cannot see $v_2$, it is a helper. 
   See Fig.~\ref{fig lem3-2} for pictorial references, whose naming will be followed during the proof.

  \begin{figure}[]
    \begin{center}
        \begin{tikzpicture}[scale=1] 
          \begin{scope}[very thick,decoration={
            markings,
            mark=at position 0.5 with {\arrow{latex}}}
            ]   
            \draw[postaction={decorate}] (-6,3) -- (0,0);
            \draw[postaction={decorate}] (-2,3) -- (0,0);
            \draw[postaction={decorate}] (2,3) -- (0,0);
            \draw[postaction={decorate}] (6,3) -- (0,0);
            
            \draw[postaction={decorate}] (-6,3) -- (-4,3);
            \draw[postaction={decorate}] (-4,3) -- (-2,3);
            \draw[postaction={decorate}] (2,3) -- (4,3);
            \draw[postaction={decorate}] (4,3) -- (6,3);
            
            \draw[postaction={decorate}] (-6,3) -- (0,6);
            \draw[postaction={decorate}] (-2,3) -- (0,6);
            \draw[postaction={decorate}] (0,6) -- (2,3);
            \draw[postaction={decorate}] (0,6) -- (6,3);

            \node[circle,draw,fill=white!20] at (0,6) {$h_1$};
            \node[circle,draw,fill=white!20] at (0,0) {$v$};
            
            \node[circle,draw,fill=white!20] at (-6,3) {$v_1$};
            \node[circle,draw,fill=white!20] at (-4,3) {$h_2$};
            \node[circle,draw,fill=white!20] at (-2,3) {$v_2$};
            \node[circle,draw,fill=white!20] at (2,3) {$v_3$};
            \node[circle,draw,fill=white!20] at (4,3) {$h_3$};
            \node[circle,draw,fill=white!20] at (6,3) {$v_4$};

            \node[] at (-2.5,4) {$F_{12}$};
            \node[] at (-2.5,2) {$F_{11}$};
            \node[] at (2.5,4) {$F_{32}$};
            \node[] at (2.5,2) {$F_{31}$};
            \node[] at (0,3) {$F_{2}$};
            \node[] at (-3,5.5) {$F_{0}$};
          \end{scope}    
          
        \end{tikzpicture}  
    
    \end{center}
    \caption{Four good vertices $v_1,v_2,v_3, v_4$ agreeing with each other
  on a vertex $v$.}
    \label{fig lem3-2}
\end{figure}
  
First suppose that $F_{11}$ is non-empty. Then the good 
vertices of $F_{11}$ must see $v_3, v_4$ via $v$, 
and thus, they must be out-neighbors of $v$. Notice that, $F_{32}$ must be empty as a good vertex from it cannot see the good vertices of $F_{11}$. On the other hand, any good vertices from $F_0, F_2$ and $F_{31}$ must see the good vertices of $F_{11}$ via $v$ by being its in-neighbor. 
Therefore, if any of $F_0, F_2$ and $F_{31}$ has a good vertex in it, then $v$ will have five good neighbors agreeing on it which is forbidden by Lemma~\ref{lem 5}. Hence the regions 
$F_0, F_2$ and $F_{31}$ are all empty. Also, as the good vertices of $F_{11}$ must be all out-neighbors of $v$, by Lemma~\ref{lem 5} we have $|F_{11}| \leq 4$. Moreover, as the vertices of $F_{11}$ cannot see the vertex $h_1$ in any way, $h_1$ functions as a helper. As $\overrightarrow{H}$ has at least $11$ good vertices, there must be at least two good vertices inside $F_{12}$. However, as the good vertices of $F_{12}$ must see $v_3,v_4$ via $h_1$ hence have no way to see $v$, then vertex $v$ also must be a helper. This forces $F_{12}$ to have at least $3$ good vertices. All these good vertices must see $v_3, v_4$ via $h_1$ which will force five vertices to agree on $h_1$, a contradiction. Therefore, the region $F_{11}$ must be empty. Similarly, one can show that the region $F_{31}$ must also be empty.

Next suppose that $F_{12}$ is non-empty.  
The good vertices from $F_{12}$ must see $v_3, v_4$ via 
$h_1$. Also, the good vertices in $F_0, F_2$, and $F_{32}$ 
must see the good vertices of $F_{12}$ via $h_1$. This 
means, except for $v$ and $h_1$, 
every good vertex is adjacent to $h_1$. 
This implies that at least $5$ good vertices 
agree on $h_1$, a contradiction to Lemma~\ref{lem 5}.
Therefore, the region $F_{12}$ must be empty.
Similarly, one can show that the region 
$F_{32}$ must also be empty.

The good vertices from the regions $F_0$ (resp., $F_{2}$) must see 
$v_2$ and $v_3$ (resp., $v_1$ and $v_4$) via $h_1$ and $v$. 
However, a good vertex from $F_0$ (resp., $F_2$) cannot see both $v_2$ and $v_3$ (resp., $v_1$ and $v_4$)   via $h_1$, 
and hence it must be adjacent to $v$ as well. 
This means,
except for $v$ and $h_1$, 
every good vertex is adjacent to $v$. 
This implies that at least $5$ good vertices 
agree on $v$, a contradiction.  
 \end{proof}
 
 \begin{figure}
    \begin{center}
        \begin{tikzpicture}[scale=0.9] 
          \begin{scope}[very thick,decoration={
            markings,
            mark=at position 0.5 with {\arrow{latex}}}
            ]   
            \draw[postaction={decorate}] (-3,3) -- (0,0);
            \draw[postaction={decorate}] (0,3) -- (0,0);
            \draw[postaction={decorate}] (3,3) -- (0,0);
            \draw[postaction={decorate}] (0,0) -- (4,0);
            
            \draw[postaction={decorate}] (0,6) -- (-3,3);
            \draw[postaction={decorate}] (3,3) -- (0,6);

            \node[circle,draw,fill=white!20] at (0,6) {$h_1$};
            \node[circle,draw,fill=white!20] at (0,0) {$v$};
            
            \node[circle,draw,fill=white!20] at (-3,3) {$v_1$};
            \node[circle,draw,fill=white!20] at (0,3) {$v_2$};
            \node[circle,draw,fill=white!20] at (3,3) {$v_3$};
            \node[circle,draw,fill=white!20] at (4,0) {$v_4$};

            \node[] at (0,4.3) {$F_{1}$};
            \node[] at (3,4.5) {$F_{0}$};
          \end{scope}    
          
        \end{tikzpicture}  
    
    \end{center}
    \caption{Three good vertices $v_1,v_2,v_3$ disagree with a fourth good vertex $v_4$ on a vertex $v$.}
    \label{fig lem3-3}
\end{figure}
 
 Using the above result, we prove the following.

 \begin{lemma}\label{lem 3,1}
 It is not possible to have three good vertices $v_1,v_2,v_3$ disagree with a fourth good vertex $v_4$ on a vertex $v$. 
 \end{lemma}
  
  \begin{proof}
  Assume that $v_1, v_2, v_3, v_4$ are arranged in a clockwise order around $v$ in the planar embedding of $\overrightarrow{H}$ and that 
  $v_1, v_2, v_3 \in N^-(v)$ while       $ v_4 \in N^+(v)$.
    Note that $v_1$ must see $v_3$ via some $h_1$. Without loss of generality, assume that $v_1 \in N^-(h_1)$ and $v_3 \in N^+(h_1)$. 
   See Fig.~\ref{fig lem3-3} for pictorial references, whose naming will be followed during the proof.

  Observe that if any good vertex, other than $v_2$, of $F_1$ see a good vertex  of $F_0$ via $v$, then $4$ good vertices will agree on $v$ which is forbidden due to Lemma~\ref{lem 4,0}. That means, every good vertex, other than $v_2$, of $F_1$ must be adjacent to $h_1$, and if $F_1$ has at least one good vertex other than $v_2$, then all good vertices of $F_0$ are also adjacent to $h_1$. That will force $h_1$ to be adjacent to at least $7$ good neighbors. That will force at least $4$ good vertices 
  to agree on $h_1$, a contradiction due to Lemma~\ref{lem 4,0}. Thus, the only good vertex in $F_1$ must be $v_2$. In particular, $|F_1| = 1$.

  Next note that, at most $3$ vertices (including $v_4$) of $F_0$ can see $v_2$ via 
  $v$, and at most $2$ vertices of $F_0$ can see $v_2$ via $h_1$ due to the forbidden structure described in Lemma~\ref{lem 4,0}. Therefore, $|F_0| \leq 5$. 
  
  That means the number of good vertices in $\overrightarrow{H}$ is
  $$|R| \leq |F_0| + |F_1| + |\{v, v_1, v_3, h_1\}| \leq 5 + 1 + 4 = 10,$$
  which is a contradiction. 
 \end{proof}
 
 Now we focus on proving that a vertex $v$ cannot have two good in-neighbors and two good out-neighbors. The proof is divided into two cases. The first case is as follows.

  \begin{lemma}\label{lem 2,2 alt}
Let $v_1,v_2,v_3, v_4$ be good neighbors of a vertex $v$ arranged in a clockwise order around $v$.   It is not possible to have  $v_1,v_3 \in N^+(v)$ and 
$v_2, v_4 \in N^-(v)$.
 \end{lemma}
  
  \begin{proof}
   Suppose that $v_1$ sees $v_3$ via some $h_1$. Then $v_2$ is forced to see $v_4$ via $h_1$ as well. This, without loss of generality, forces the embedding presented in  Fig.~\ref{fig lem3-6}, whose naming will be followed during the proof.

  \begin{figure}
    \begin{center}
        \begin{tikzpicture}[scale=1] 
          \begin{scope}[very thick,decoration={
            markings,
            mark=at position 0.5 with {\arrow{latex}}}
            ]   
            \draw[postaction={decorate}] (0,0) -- (-6,3);
            \draw[postaction={decorate}] (-2,3) -- (0,0);
            \draw[postaction={decorate}] (0,0) -- (2,3);
            \draw[postaction={decorate}] (6,3) -- (0,0);
            
            
            \draw[] (-6,3) -- (0,6);
            \draw[] (-2,3) -- (0,6);
            \draw[] (0,6) -- (2,3);
            \draw[] (0,6) -- (6,3);

            \node[circle,draw,fill=white!20] at (0,6) {$h_1$};
            \node[circle,draw,fill=white!20] at (0,0) {$v$};
            
            \node[circle,draw,fill=white!20] at (-6,3) {$v_1$};
            \node[circle,draw,fill=white!20] at (-2,3) {$v_2$};
            \node[circle,draw,fill=white!20] at (2,3) {$v_3$};
            \node[circle,draw,fill=white!20] at (6,3) {$v_4$};

            \node[] at (-4,3) {$F_{1}$};
            \node[] at (0,3) {$F_{2}$};            
            \node[] at (4,3) {$F_{3}$};
            \node[] at (-3,5.5) {$F_{0}$};
          \end{scope}    
          
        \end{tikzpicture}  
    
    \end{center}
    \caption{The good neighbors $v_1,v_3 \in N^+(v)$ and $v_2, v_4 \in N^-(v)$.}
    \label{fig lem3-6}
\end{figure}

  Any good vertex of $F_i$ must see $v_{i+2}$  and $v_{i+3}$ via either $v$ or $h_1$, where $i \in \{0,1,2,3\}$ and the $+$ operation is taken modulo $4$.
  Note that, there are at least $5$ more vertices in $\bigcup F_i$, i.e., $|F_0|+|F_1|+|F_2|+|F_3| \geq 5$.
  This will force $v$ or $h_1$ to become adjacent to $7$ 
  good vertices among which three disagree with the fourth one on $v$ or $h_1$, respectively. 
    This is a contradiction to Lemma~\ref{lem 3,1}.
 \end{proof}

Now we present the second case.

  \begin{lemma}\label{lem 2,2 cons}
Let $v_1,v_2,v_3, v_4$ be good neighbours of a vertex $v$ arranged in a clockwise order around $v$.   It is not possible to have  $v_1,v_2 \in N^+(v)$ and 
$v_3, v_4 \in N^-(v)$.
 \end{lemma}
  
  \begin{proof}
Notice that $v_1$ must see $v_2$ via some $h_1$. 
Also assume that $v_3$ sees $v_4$ via the same $h_1$, see Fig.~\ref{fig lem3-5}(a).
Any good vertex $v'$ other than $v,v_1,v_2,v_3,v_4,h_1$ is separated from $v_i$, for some $i\in \{1,2,3,4\}$, by the $4$-cycle of the form $vv_jh_1v_kv$, where $j,k \in \{1,2,3,4\}$. Then either $v$ or $h_1$ is forced to be adjacent to four good vertices, among which three disagrees with the fourth one on $v$ or $h_1$. This contradicts Lemma~\ref{lem 3,1}.
Thus $v_3$ must see $v_4$ via a different vertex $h_2$. 
See Fig.~\ref{fig lem3-5}(b) for pictorial references, whose naming will be followed during the proof.

Note that a good vertex in $F_i$ where $i\in \{0,1,2\}$ is forced to not see any other good vertex via $v$ due to Lemma~\ref{lem 3,1}.
Thus, any good vertex of $F_1$ must see $v_3$ and $v_4$ via $h_1$, which is again a contradiction to Lemma~\ref{lem 3,1} at $h_1$. Hence $F_1$ does not contain a good vertex. Similarly, $F_2$ also doesn't contain any good vertices.
Therefore $F_0$ is non-empty.

Let $W = R \setminus   \{v,v_1,v_2,v_3,v_4,h_1,h_2\} = \{w_1, w_2, \cdots, w_t\}$. Observe that $|W| = t \geq 11 - 7 = 4$ and the vertices of $W$ must be contained in $F_0$.  Each vertex of $W$ can see at most two of the four good vertices $v_1,v_2,v_3,v_4$ via a particular vertex due to Lemma~\ref{lem 3,1}.       
    Let us assume that the separating cycle by the arcs $v_2v,vv_3$ and the paths used by $w_1$ to see $v_2$ and $v_3$ is $C$. Without loss of generality, we may assume that the vertices from $W\setminus \{ w_1 \}$ are not separated from $v_1$ and $v_4$ by $C$. 
    Let's examine a few possible cases.
    
    \begin{itemize}
    \item  If $w_1$ sees $v_1, v_2$ via $h_3$ and $v_3,v_4$ via $h_4$, then both $w_2$ and $w_3$ are forced to see 
    $ v_2$ via $h_3$. Thus $h_3$ has three good neighbours disagreeing on it with a fourth one. 
    This is a contradiction to Lemma~\ref{lem 3,1}. 
    
     \item  If $w_1$ sees $v_1, v_3$ via $h_3$ and $v_2,v_4$ via $h_4$, then the graph $\overrightarrow{H}$ no longer remains planar. 
    
  \item  If $w_1$ sees $v_1, v_2$ via $h_3$ and $v_3$ via $h_4$ (or by being adjacent to it) and $v_4$ via $h_5$ (or by being adjacent to it), then $w_2$ is not able to see $v_3$.
  
   \item  If $w_1$ sees $v_1$ via $h_3$
   (or by being adjacent to it)  and $ v_2$ via $h_4$ (or by being adjacent to it) and $v_3$ via $h_5$ (or by being adjacent to it)
   and $v_4$ via $h_6$ (or by being adjacent to it), then $w_2$ is not able to see $v_2,v_3$.
        \end{itemize}

    Thus, we are done.    
    \end{proof}

\begin{figure}
     \begin{subfigure}{0.5\textwidth}
    \begin{center}
      \begin{tikzpicture}[scale=0.6] 
         \begin{scope}[very thick,decoration={
          markings,
          mark=at position 0.5 with {\arrow{latex}}}
          ]   
          \draw[postaction={decorate}] (-5,3) -- (0,0);
          \draw[postaction={decorate}] (-1,3) -- (0,0);
          \draw[postaction={decorate}] (0,0) -- (1,3);
          \draw[postaction={decorate}] (0,0) -- (5,3);
          
          \draw[thick]  (-3,3) -- (-5,3);
          \draw[thick] (-3,3) -- (-1,3);
          \draw[thick] (1,3) -- (3,3);
          \draw[thick] (3,3) -- (5,3);

          \draw[thick] (0,6) -- (-5,3);
          \draw[thick] (0,6) -- (-1,3);
          \draw[thick] (0,6) -- (1,3);
          \draw[thick] (0,6) -- (5,3);
          

          \node[circle,draw,fill=white!20] at (0,0) {$v$};
          
          \node[circle,draw,fill=white!20] at (-5,3) {$v_1$};
          \node[circle,draw,fill=white!20] at (0,6) {$h_1$};
          \node[circle,draw,fill=white!20] at (-1,3) {$v_2$};
          \node[circle,draw,fill=white!20] at (1,3) {$v_3$};
          \node[circle,draw,fill=white!20] at (5,3) {$v_4$};

        \end{scope}    
        
      \end{tikzpicture}   
  
    \end{center}
    \subcaption{}
    \label{fig lem3-5-a}
  \end{subfigure}
  \begin{subfigure}{0.5\textwidth}
    \begin{center}
      \begin{tikzpicture}[scale=0.6] 
        \begin{scope}[very thick,decoration={
          markings,
          mark=at position 0.5 with {\arrow{latex}}}
          ]   
          \draw[postaction={decorate}] (-5,3) -- (0,0);
          \draw[postaction={decorate}] (-1,3) -- (0,0);
          \draw[postaction={decorate}] (0,0) -- (1,3);
          \draw[postaction={decorate}] (0,0) -- (5,3);
          
          \draw[thick]  (-3,3) -- (-5,3);
          \draw[thick] (-3,3) -- (-1,3);
          \draw[thick] (1,3) -- (3,3);
          \draw[thick] (3,3) -- (5,3);
          

          \node[circle,draw,fill=white!20] at (0,0) {$v$};
          
          \node[circle,draw,fill=white!20] at (-5,3) {$v_1$};
          \node[circle,draw,fill=white!20] at (-3,3) {$h_1$};
          \node[circle,draw,fill=white!20] at (-1,3) {$v_2$};
          \node[circle,draw,fill=white!20] at (1,3) {$v_3$};
          \node[circle,draw,fill=white!20] at (3,3) {$h_2$};
          \node[circle,draw,fill=white!20] at (5,3) {$v_4$};

          \node[] at (0,2) {$F_0$};
          \node[] at (-2,2) {$F_1$};
          \node[] at (2,2) {$F_2$};
          
        \end{scope}    
        
      \end{tikzpicture}  
  
      \end{center}
      
    \subcaption{}
    \label{fig lem3-5-b}
  \end{subfigure}

  \caption{The good neighbors $v_1,v_2 \in N^+(v)$ and $v_3, v_4 \in N^-(v)$, (a) $v_1$ sees $v_2$ via $h_1$ and $v_3$ sees $v_4$ via the same $h_1$, (b) $v_1$ sees $v_2$ via $h_1$ and $v_3$ sees $v_4$ via some $h_2$}
  
  \label{fig lem3-5}
\end{figure}

\begin{lemma}\label{lem 4}
 It is not possible for a vertex to have at least four good neighbours.  
\end{lemma}

\begin{proof}
    Follows directly from Lemmata~\ref{lem 4,0},~\ref{lem 3,1},~\ref{lem 2,2 alt} and~\ref{lem 2,2 cons}.
\end{proof}

 \subsection{Forbidding three good vertices in a neighbourhood}
In this subsection, we want to improve the results 
of the previous subsection. 
Indeed, the previous results will be used to prove so. 

\begin{figure}
  \begin{subfigure}{0.5\textwidth}
    \begin{center}
      \begin{tikzpicture}[scale=0.8] 
        \begin{scope}[very thick,decoration={
          markings,
          mark=at position 0.5 with {\arrow{latex}}}
          ]   
          \draw[postaction={decorate}] (0,0) -- (-4,3);
          \draw[postaction={decorate}] (0,0) -- (0,3);
          \draw[postaction={decorate}] (0,0) -- (4,3);
          
          
          \draw[] (-4,3) -- (0,6);
          \draw[] (0,3) -- (0,6);
          \draw[] (0,6) -- (4,3);

          \node[circle,draw,fill=white!20] at (0,6) {$h_1$};
          \node[circle,draw,fill=white!20] at (0,0) {$v$};
          
          \node[circle,draw,fill=white!20] at (-4,3) {$v_1$};
          \node[circle,draw,fill=white!20] at (0,3) {$v_2$};
          \node[circle,draw,fill=white!20] at (4,3) {$v_3$};

          \node[] at (-2,3) {$F_{1}$};
          \node[] at (2,3) {$F_{2}$};
          \node[] at (-3,5.5) {$F_{0}$};
        \end{scope}    
        
      \end{tikzpicture}   
  
    \end{center}
    \subcaption{}
    \label{fig 3agree-a}
  \end{subfigure}
  \begin{subfigure}{0.5\textwidth}
    \begin{center}
      \begin{tikzpicture}[scale=0.8] 
        \begin{scope}[very thick,decoration={
          markings,
          mark=at position 0.5 with {\arrow{latex}}}
          ]   
          \draw[postaction={decorate}] (0,0) -- (-4,3);
          \draw[postaction={decorate}] (0,0) -- (0,3);
          \draw[postaction={decorate}] (0,0) -- (4,3);
          
          
          \draw[] (-4,3) -- (0,6);
          \draw[] (0,6) -- (4,3);
          \draw[] (-4,3) -- (-2,3);
          \draw[] (-2,3) -- (0,3);
          \draw[] (0,3) -- (2,3);
          \draw[] (2,3) -- (4,3);

          \node[circle,draw,fill=white!20] at (0,6) {$h_3$};
          \node[circle,draw,fill=white!20] at (-2,3) {$h_1$};
          \node[circle,draw,fill=white!20] at (2,3) {$h_2$};

          \node[circle,draw,fill=white!20] at (0,0) {$v$};
          
          \node[circle,draw,fill=white!20] at (-4,3) {$v_1$};
          \node[circle,draw,fill=white!20] at (0,3) {$v_2$};
          \node[circle,draw,fill=white!20] at (4,3) {$v_3$};

          \node[] at (-1,2) {$F_{1}$};
          \node[] at (1,2) {$F_{2}$};
          \node[] at (0,4.5) {$F_{3}$};
          \node[] at (-3,5.5) {$F_{0}$};
        \end{scope}    
        
      \end{tikzpicture}   
  
    \end{center}
    \subcaption{}
    \label{fig 3agree-b}
  \end{subfigure}
  
  \caption{Three good vertices $v_1,v_2,v_3$ agreeing on another good vertex $v$.}
  
  \label{fig 3agree}
\end{figure}

\begin{lemma}\label{lem 3,0}
  Three good vertices $v_1,v_2,v_3$ cannot agree with each other on a vertex $v$.
\end{lemma}

\begin{proof}

  Without loss of generality, suppose that $v_1,v_2,v_3$ are arranged in a clockwise order around $v$ in the planar embedding of $\overrightarrow{H}$ such that $v_1,v_2,v_3$ are out-neighbours of $v$. 
  Let $v_1$ see $v_2$ via some $h_1$. Let $v_3$ also see $v_2$ via $h_1$. See Fig.~\ref{fig 3agree}(a) for pictorial references, whose naming will be followed during the proof.
  
  Any good vertices contained in $F_i$, where $i\in \{0,1,2\}$ must see $v_{i+2}$ ($+$ is taken modulo $3$) via $v$ or $h_1$. This forces $v$ or $h_1$ to have $4$ good vertices in their neighbourhood, which we proved to be forbidden in Lemma~\ref{lem 4}.
  Hence, $v_2$ must see $v_3$ via  $h_2$ and $v_3$ must see $v_1$ via  $h_3$. See Fig.~\ref{fig 3agree}(b) for pictorial references.
  Any good vertex in regions $F_i,i\in \{0,1,2\}$ sees $v_{i+2}$ ($+$ is taken modulo $3$) via $v$ or $h_i$, which contradicts Lemma~\ref{lem 4}.
  Hence $F_0, F_1, F_2$ are empty.

  Let  $W=R \setminus \{v, v_1,v_2,v_3,h_1,h_2,h_3\} = \{ w_1, w_2, \dots, w_t \}$ and, thus, $|W| = t \geq 11 - 7 = 4$. 
  Suppose that $w_1$ sees $v_1,v_3$ via the same vertex, say, $h_4$, and 
  the cycle $C_1 = v_1h_3v_3h_4v_1$ separates all vertices of 
  $W \setminus \{w_1\}$ from $v_2$. Then, all vertices of $W \setminus \{w_1\}$ must see $v_2$ via $h_4$, which will force $h_4$ to have at least $4$ good neighbors, contradicting Lemma~\ref{lem 4}.  Therefore, $w_1$ must have two internally vertex disjoint paths to see $v_1$ and $v_3$, respectively.   
  Let us assume that $C$ is the cycle induced by the arcs $v_1h_3,h_3v_3$ and the paths used by $w_1$ to see $v_1$ and $v_3$. Without loss of generality, we may assume that the vertices from $W\setminus \{ w_1 \}$ are  separated from $v_2$ by $C$.
  If $w_1$ sees $v_1$ (resp., $v_3$) via a vertex, then call it $h_4$ (resp., $h_5$).
 Note that, the only options for $w_2$ to see $v_2$ is via $h_4, h_5$, (if they exist) or $w_1$ (if $w_1$ is adjacent to $v_2$). In any of these cases, we contradict Lemma~\ref{lem 4} as a vertex ($h_4, h_5$ or $w_1$) having at least $4$ good vertices is forced. 
 \end{proof}

The final lemma in the similar direction follows.

\begin{lemma}\label{lem 2,1}
  It is not possible to have  two good vertices $v_1,v_2$ disagree with a third good vertex $v_3$ on a vertex $v$. 
\end{lemma}
  
\begin{proof}
   Without loss of generality, suppose that $v_1,v_2,v_3$ are arranged in a clockwise order around $v$ in the planar embedding of $\overrightarrow{H}$ such that $v_1,v_2$ are in-neighbours of $v$ and $v_3$ 
   is an out-neighbour of $v$. Note that $v_1$ must see $v_2$ via some vertex $h_1$. Without loss of generality assume that the vertex $v_3$ 
  is placed in the exterior of the  closed curve induced by the cycle $vv_1h_1v_2$.  See Fig.~\ref{fig 2disagree1} for pictorial references, whose naming will be followed during the proof.

  Any good vertex contained in $F_1$ must see $v_3$ via  $v$ or $h_1$. This will force $v$ or $h_1$ to have at least four good neighbors contradicting  Lemma~\ref{lem 4}. 
  Thus $F_1$ is empty. 

  Let $W=R \setminus \{v, v_1,v_2,v_3,h_1\} = \{ w_1, w_2, \dots, w_t \}$ and thus, $|W| = t \geq 11 - 5 =  6$.
  Notice that $h_1$ is already adjacent to $2$ good vertices, namely, $v_1$ and $v_2$. 
  Thus, due to Lemma~\ref{lem 4} we know that at most one vertex from $W$ can be adjacent to $h_1$. Thus, without loss of generality we may assume that all vertices from $W \setminus \{w_1\}$ are non-adjacent to $h_1$. 

  Suppose that $w_2$ sees $v_1$ and $v_3$ via 
  a vertex $h_2$ and that the cycle $C_1 = vv_1h_2v_3v$ 
  separates the vertices of $W \setminus \{w_1, w_2\}$ from $v_2$. Thus, the vertices of $W \setminus \{w_1, w_2\}$ must see $v_2$ via $h_2$, forcing $h_2$ 
  to have at least $4$ good neighbors contradicting Lemma~\ref{lem 4}. 
  Therefore, $w_2$ must have two internally vertex disjoint paths to see $v_1$ and $v_3$, respectively.   
  Let us assume that $C$ is the cycle induced by the arcs $vv_1,vv_3$ and the paths used by $w_2$ to see $v_1$ and $v_3$. Without loss of generality, we may assume that the vertices from 
  $W \setminus \{w_1, w_2\}$ are  separated from $v_2$ by $C$. 

  Assume that $w_2$ sees $v_1, v_2$ via $h_1, h_3$, respectively. 
  The cycle $C = vv_2h_1w_1h_3v_3v$ separates all the vertices of $W \setminus \{w_1$\} from $v_2$.
If $w_2$ sees $v_1$ (resp., $v_3$) via some vertex, then let that vertex be $h_2$ (resp., $h_3$). 
Note that, the only options for $w_3$ to see $v_2$ is via $h_2, h_3$, (if they exist) or $w_2$ (if $w_2$ is adjacent to $v_2$). In any of these cases, we contradict Lemma~\ref{lem 4} as a vertex ($h_2, h_3$ or $w_2$) having at least $4$ good vertices is forced. 
\end{proof}

This implies  that the graph $\overrightarrow{H}$ does not have any vertex $v$ with at least three good neighbors. 

\begin{figure}
  \begin{center}
      \begin{tikzpicture}[scale=0.8] 
        \begin{scope}[very thick,decoration={
          markings,
          mark=at position 0.5 with {\arrow{latex}}}
          ]   
          \draw[postaction={decorate}] (-3,3) -- (0,0);
          \draw[postaction={decorate}] (3,3) -- (0,0);
          \draw[postaction={decorate}] (0,0) -- (4,0);
          
          \draw[] (0,6) -- (-3,3);
          \draw[] (3,3) -- (0,6);

          \node[circle,draw,fill=white!20] at (0,6) {$h_1$};
          \node[circle,draw,fill=white!20] at (0,0) {$v$};
          
          \node[circle,draw,fill=white!20] at (-3,3) {$v_1$};
          \node[circle,draw,fill=white!20] at (3,3) {$v_2$};
          \node[circle,draw,fill=white!20] at (4,0) {$v_3$};

          \node[] at (0,3) {$F_{1}$};
          \node[] at (3,4.5) {$F_{0}$};
        \end{scope}    
        
      \end{tikzpicture}  
  
  \end{center}
  \caption{Two good vertices $v_1,v_2$ disagree with a third good vertex $v_3$ on a vertex $v$.}
  \label{fig 2disagree1}
\end{figure}

\begin{lemma}\label{lem 3}
 It is not possible for a vertex to have at least three good neighbors.  
 \end{lemma}
  
  \begin{proof}
Follows directly from Lemmata~\ref{lem 3,0} and~\ref{lem 2,1}.  
 \end{proof}

\subsection{The final part of the proof}
In particular, for any helper $h$ we have $d(h) \leq 2$. Thus, using our earlier observation that the degree of $h$ is at least $2$, we can conclude that $d(h)=2$. 

Observe that two helpers $h_1$ and $h_2$ cannot have $N(h_1) = N(h_2) = \{u,v\}$. The reason is, both $h_1$ and $h_2$  contribute in  $u$ seeing $v$. Therefore, even if we delete $h_2$, the set $R$ still remains an oriented relative clique contradicting the minimality of $\overrightarrow{H}$.

Now construct a graph $\overrightarrow{H}^*$ from $\overrightarrow{H}$ as follows: delete each helper and add an edge between its neighbors. Observe that $\overrightarrow{H}^*$ is planar, not necessarily triangle-free, with $V(H^*)$ being the set of good neighbors of 
$\overrightarrow{H}$. 
Also, the degree of a vertex $v$ in $\overrightarrow{H}^*$ is greater than or equal to the 
degree of $v$ in $\overrightarrow{H}$. 
As $\overrightarrow{H}^*$ is a planar graph, it must have a vertex $x$ with degree at most five. Therefore, we can say that there exists a good vertex $x$ in 
$\overrightarrow{H}$ having degree at most five. We fix the name of this vertex
$x$ for the rest of this section.

\medskip

\noindent \textit{Proof of Theorem~\ref{th main}.}  
Let $x$ be a good vertex of $\overrightarrow{H}$ 
having degree at most five whose existence follows 
from the above paragraph. 
Let $X = R \setminus (N(x) \cup \{x\})$. 
We know due to Lemma~\ref{lem 3} that 
$|R \cap N(x)| \leq 2$.

As  $|R| \geq 11$, we have $|X| \geq 8$.  
Note that each vertex of $X$ must see $x$ via one of
its neighbors.
Therefore, by pigeonhole principle at least one of 
the neighbors $x_1$  
of $x$ will have two good neighbors from $X$. Thus 
$x_1$ has three good neighbors, contradicting 
Lemma~\ref{lem 3}. \qed

\section{Properties of minimal bounds of $\mathcal{P}_4$}\label{sec min_bounds}
To begin this section, we will prove that if $\overrightarrow{T}$ is a
minimal bound of $\mathcal{P}_4$, that is, the family of triangle-free planar graphs,
then it is necessary for every vertex  of $T$ to have degree at least $10$.

\begin{theorem}\label{th P4 universal min deg}
    Let $\overrightarrow{T}$ be a minimal bound of 
    $\mathcal{P}_4$, that is, the family of triangle-free planar graphs. 
    Then $d(v) = d^+(v) + d^-(v) \geq 10$ 
    for all $v \in V(\overrightarrow{T})$. 
\end{theorem}

\begin{proof}
Let $\overrightarrow{T}$ be a minimal bound of $\mathcal{P}_4$, that is, the family of triangle-free planar graphs.  Thus, there exists a $\overrightarrow{H} \in \mathcal{P}_4$ such that for any homomorphism $f: \overrightarrow{H} \to \overrightarrow{T}$  the following are satisfied:
\begin{enumerate}[(a)]
    \item For any $v \in V(\overrightarrow{T})$, there exists some $v' \in V(\overrightarrow{H})$ such that $f(v') = v$.

    \item For any $uv \in A(\overrightarrow{T})$, there exists some $u'v' \in A(\overrightarrow{H})$ such that $f(u') = u$ and $f(v') = v$.
\end{enumerate}

\begin{figure}
  \begin{center}
      \begin{tikzpicture}[scale=0.5] 
        \begin{scope}[very thick,decoration={
          markings,
          mark=at position 0.5 with {\arrow{latex}}}
          ]   


          \node[circle,draw,fill=white!20,inner sep=3pt,minimum size=10pt] (y) at (0,0) {\footnotesize$y$};
          \node[circle,draw,fill=white!20,inner sep=3pt,minimum size=10pt] (x) at (0,8) {\footnotesize$x$};

          \node[circle,draw,fill=white!20,inner sep=1pt,minimum size=10pt] (v1) at (-9.5,4) {\footnotesize$v_1$};
          \node[rectangle,draw,fill=white!20,inner sep=2pt,minimum size=10pt] (h1) at (-7.5,4) {\footnotesize$h_1$};
          \node[circle,draw,fill=white!20,inner sep=1pt,minimum size=10pt] (v2) at (-6,4) {\footnotesize$v_2$};
          \node[circle,draw,fill=white!20,inner sep=1pt,minimum size=10pt] (v3) at (-4,4) {\footnotesize$v_3$};
          \node[rectangle,draw,fill=white!20,inner sep=2pt,minimum size=10pt] (h2) at (-2.5,4) {\footnotesize$h_2$};
          \node[circle,draw,fill=white!20,inner sep=1pt,minimum size=10pt] (v4) at (-1,4) {\footnotesize$v_4$};
          \node[circle,draw,fill=white!20,inner sep=1pt,minimum size=10pt] (v5) at (1,4) {\footnotesize$v_5$};
          \node[rectangle,draw,fill=white!20,inner sep=2pt,minimum size=10pt] (h3) at (2.5,4) {\footnotesize$h_3$};
          \node[circle,draw,fill=white!20,inner sep=1pt,minimum size=10pt] (v6) at (4,4) {\footnotesize$v_6$};
          \node[circle,draw,fill=white!20,inner sep=1pt,minimum size=10pt] (v7) at (6,4) {\footnotesize$v_7$};
          \node[rectangle,draw,fill=white!20,inner sep=2pt,minimum size=10pt] (h4) at (7.5,4) {\footnotesize$h_4$};
          \node[circle,draw,fill=white!20,inner sep=1pt,minimum size=10pt] (v8) at (9.5,4) {\footnotesize$v_8$};

          \node[] (=) at (12,4) {\Large$=$};

          \node[circle,draw,fill=white!20,inner sep=3pt,minimum size=10pt] (y1) at (14.5,0) {\footnotesize$y$};
          \node[circle,draw,fill=white!20,inner sep=3pt,minimum size=10pt] (x1) at (14.5,8) {\footnotesize$x$};


          \draw[postaction={decorate}]  (v1) -- (x);
          \draw[postaction={decorate}]  (v2) -- (x);
          \draw[postaction={decorate}]  (v3) -- (x);
          \draw[postaction={decorate}]  (v4) -- (x);

          \draw[postaction={decorate}]  (x) -- (v5);
          \draw[postaction={decorate}]  (x) -- (v6);
          \draw[postaction={decorate}]  (x) -- (v7);
          \draw[postaction={decorate}]  (x) -- (v8);

          \draw[postaction={decorate}]  (v1) -- (y);
          \draw[postaction={decorate}]  (v2) -- (y);
          \draw[postaction={decorate}]  (y) -- (v3);
          \draw[postaction={decorate}]  (y) -- (v4);

          \draw[postaction={decorate}]  (y) -- (v5);
          \draw[postaction={decorate}]  (y) -- (v6);
          \draw[postaction={decorate}]  (v7) -- (y);
          \draw[postaction={decorate}]  (v8) -- (y);

          \draw (v1) -- (h1) -- (v2);
          \draw (v3) -- (h2) -- (v4);
          \draw (v5) -- (h3) -- (v6);
          \draw (v7) -- (h4) -- (v8);

          \draw[very thick,red] (x1) -- (y1);

        \end{scope}    
        
      \end{tikzpicture}  
  
  \end{center}
  \caption{The oriented triangle-free graph $\overrightarrow{G}_0$.}
  \label{fig gadget}
\end{figure}

Next we are going to build a gadget graph recursively 
for our proof using the triangle-free planar oriented 
graph $\overrightarrow{G}_0$  depicted in Fig.~\ref{fig gadget} 
as the base graph. 
Notice that, in $\overrightarrow{G}_0$, the vertex subset $|R|= \{x, y, v_1, v_2, \ldots, v_8\}$ is 
an oriented relative clique of order $10$. 
As $\overrightarrow{G}_0$ is an
oriented triangle-free planar  graph, it must admit a homomorphism $g_0$ (say) to $\overrightarrow{T}$. 
An important observation to make here is that the 
images $g_0(x), g_0(y)$ of $x, y$, respectively, 
will also have at least two vertices (common neighbors) in the set 
$N^\alpha(g_0(x)) \cap N^\beta(g_0(y))$, for all 
$\alpha, \beta \in \{+,-\}$ inside $\overrightarrow{T}$. As this property is important for the proof, let us define it formally. 
Two distinct vertices $u,v$ of $\overrightarrow{T}$ 
are in a $X_2$-relation  if 
we have $ |N^\alpha(u) \cap N^\beta(v)| \geq 2$
for all $\alpha,\beta \in \{+,-\}$.

Let us now continue to build our gadget graph further. 
Recall that there exists an oriented triangle-free outerplanar graph $\overrightarrow{M}$ with oriented chromatic number $6$, constructed by  Dolama~\cite{hosse2005}. 
We construct a mixed graph by adding a vertex $z$ 
to the oriented graph $\overrightarrow{M}$ which is adjacent to each vertex of $\overrightarrow{M}$ via a (red) edge (depicted in Fig.~\ref{fig H}). Now replace each of the (red) edges $zw$ with the oriented graph $\overrightarrow{G}_0$. 
The precise step-wise process of this replacement is: 
take a copy of $\overrightarrow{G}_0$, identify its 
vertex $x$ (resp., $y$) with the vertex $z$ 
(resp., $w$) of the mixed graph, 
delete the (red) edge $zw$. The so-obtained oriented graph is called $\overrightarrow{G}_1$. 
Observe that, $\overrightarrow{G}_1$ is an 
oriented triangle-free planar graph, and thus it must admit a homomorphism $g_1$ (say) to $\overrightarrow{T}$.
As the oriented chromatic number of $\overrightarrow{M}$ is $6$, and as due to our construction $g_1(z)$ must be in $X_2$-relation with 
every vertex of $\overrightarrow{M}$ (sitting inside $\overrightarrow{G}_1$ as a subgraph), the image $g_1(z)$ in $\overrightarrow{T}$ must be in $X_2$-relation with at least $6$ other vertices of $\overrightarrow{T}$.

After this, take $|V(\overrightarrow{H})|$ copies of $\overrightarrow{G}_1$ and for each vertex $v$ of $\overrightarrow{H}$, take one copy of 
$\overrightarrow{G}_1$ and identify its vertex $z$ to the vertex $v$. This so-obtained oriented graph is called $\overrightarrow{G}_2$. 
Observe that, as $\overrightarrow{G}_2$ is an 
oriented triangle-free planar  graph, 
it must admit a homomorphism $g_2$ (say) to $\overrightarrow{T}$.
As each vertex of $\overrightarrow{H}$ (sitting inside $\overrightarrow{G}_2$ as a subgraph) plays the role of $z$ in $\overrightarrow{G}_1$, and as any homomorphism of $\overrightarrow{H}$ to $\overrightarrow{T}$ is onto (on the set of vertices), each vertex of $\overrightarrow{T}$ must be in 
$X_2$-relation with at least $6$ other vertices.

We are now ready to conclude the proof in the method of contradiction. Thus, let us assume that $v$ is a vertex of $\overrightarrow{T}$ satisfying 
$d(v) = d^{+}(v) + d^{-}(v) \leq 9$. 
First of all note that $d^+(v), d^-(v) \geq 4$ 
in order for $v$ to have 
$X_2$-relation with another vertex. 
Therefore, for some $\alpha \in \{+,-\}$ 
we must have $d^{\alpha}(v) = 4$. 
As $v$ has $X_2$-relation with at least $6$ vertices, 
by the pigeonhole principle, 
$v$ must have $X_2$-relation with another vertex, suppose $u$, from $N^{\alpha}(v)$. This will imply 
that we have at least two vertices in the set 
$A = N^+(u) \cap N^{\alpha}(v)$ and at least two vertices in the set $B = N^-(u) \cap N^{\alpha}(v)$. 
Notice that $A \cap B = \emptyset$ as 
$N^+(u) \cap N^-(u) = \emptyset$.

\begin{figure}
  \begin{center}
      \begin{tikzpicture}[scale=0.4] 
        \begin{scope}[very thick,decoration={
          markings,
          mark=at position 0.5 with {\arrow{latex}}}
          ]   


          \node[circle,draw,fill=white!20,inner sep=2pt,minimum size=10pt] (a) at (90:3) {};
          \node[circle,draw,fill=white!20,inner sep=2pt,minimum size=10pt] (b) at (18:3) {};
          \node[circle,draw,fill=white!20,inner sep=2pt,minimum size=10pt] (c) at (306:3) {};
          \node[circle,draw,fill=white!20,inner sep=2pt,minimum size=10pt] (d) at (234:3) {};
          \node[circle,draw,fill=white!20,inner sep=2pt,minimum size=10pt] (e) at (162:3) {};

          \node[circle,draw,fill=white!20,inner sep=2pt,minimum size=10pt] (a1) at (110:5) {};
          \node[circle,draw,fill=white!20,inner sep=2pt,minimum size=10pt] (a2) at (70:5) {};

          \node[circle,draw,fill=white!20,inner sep=2pt,minimum size=10pt] (e1) at (182:5) {};
          \node[circle,draw,fill=white!20,inner sep=2pt,minimum size=10pt] (e2) at (142:5) {};

          \node[circle,draw,fill=white!20,inner sep=2pt,minimum size=10pt] (d1) at (254:5) {};
          \node[circle,draw,fill=white!20,inner sep=2pt,minimum size=10pt] (d2) at (214:5) {};

          \node[circle,draw,fill=white!20,inner sep=2pt,minimum size=10pt] (c1) at (326:5) {};
          \node[circle,draw,fill=white!20,inner sep=2pt,minimum size=10pt] (c2) at (286:5) {};

          \node[circle,draw,fill=white!20,inner sep=2pt,minimum size=10pt] (b1) at (38:5) {};
          \node[circle,draw,fill=white!20,inner sep=2pt,minimum size=10pt] (b2) at (-2:5) {};

          \node[circle,draw,fill=white!20,inner sep=2pt,minimum size=10pt] (z1) at (90:5) {};
          \node[circle,draw,fill=white!20,inner sep=2pt,minimum size=10pt] (z2) at (80:7) {};

          \node[circle,draw,fill=white!20,inner sep=2pt,minimum size=10pt] (z) at (360:10) {\footnotesize $z$};


          \draw[-latex]  (a) -- (b);
          \draw[-latex]  (b) -- (c);
          \draw[-latex]  (c) -- (d);
          \draw[-latex]  (d) -- (e);
          \draw[-latex]  (e) -- (a);

          \draw[-latex]  (a1) -- (a);
          \draw[-latex]  (a) -- (a2);
          \draw[-latex]  (a2) -- (b1);
          \draw[-latex]  (b1) -- (b);
          \draw[-latex]  (b) -- (b2);
          \draw[-latex]  (b2) -- (c1);
          \draw[-latex]  (c1) -- (c);
          \draw[-latex]  (c) -- (c2);
          \draw[-latex]  (c2) -- (d1);
          \draw[-latex]  (d1) -- (d);
          \draw[-latex]  (d) -- (d2);
          \draw[-latex]  (d2) -- (e1);
          \draw[-latex]  (e1) -- (e);
          \draw[-latex]  (e) -- (e2);
          \draw[-latex]  (e2) -- (a1);
          \draw[-latex]  (a) -- (z1);
          \draw[-latex]  (z1) -- (z2);
          \draw[-latex]  (z2) -- (a2);

          \begin{pgfonlayer}{bg}
            \draw[thick,red!40] (z) -- (a);
            \draw[thick,red!40] (z) -- (b);
            \draw[thick,red!40] (z) -- (c);
            \draw[thick,red!40] (z) -- (d);
            \draw[thick,red!40] (z) -- (e);
            \draw[thick,red!40] (z) -- (a1);
            \draw[thick,red!40] (z) -- (a2);
            \draw[thick,red!40] (z) -- (b1);
            \draw[thick,red!40] (z) -- (b2);
            \draw[thick,red!40] (z) -- (c1);
            \draw[thick,red!40] (z) -- (c2);
            \draw[thick,red!40] (z) -- (d1);
            \draw[thick,red!40] (z) -- (d2);
            \draw[thick,red!40] (z) -- (e1);
            \draw[thick,red!40] (z) -- (e2);
            \draw[thick,red!40] (z) -- (z1);
            \draw[thick,red!40] (z) -- (z2);
          \end{pgfonlayer}

        \end{scope}    
        
      \end{tikzpicture}  
  
  \end{center}
  \caption{The oriented triangle-free outerplanar graph $\overrightarrow{M}$ having oriented chromatic number $6$ whose each vertex is adjacent to an additional vertex $z$ with a (red) edge.}
  \label{fig H}
\end{figure}

Thus, notice that, 
$$4 = d^{\alpha}(v) = |N^{\alpha}(v)| \geq |A| +|B| + |\{u\}| \geq 2 + 2 + 1 = 5,$$
a contradiction. Therefore, $d(v) \geq 10$ for all $v \in V(\overrightarrow{T})$.  
\end{proof}

\subsection{Application 1: oriented chromatic number}\label{subsec app1}
Recall that the bound $\chi_o(\mathcal{P}_4) \geq 11$
was achieved by Ochem~\cite{Ochemgirth4} using a computer check. Here, as a corollary of 
Theorem~\ref{th P4 universal min deg} we obtain a theoretical proof of the same, and slightly improve the result. 

According to Remark~\ref{remark oriented chromatic minimal family bounds}, we know that there exists a minimal bound of $\mathcal{P}_4$ on $\chi_o(\mathcal{P}_4)$ vertices. However, for the sake of completeness, and also to facilitate similar proofs in future, let us prove a lemma in a generalized set-up.

A  family $\mathcal{F}$ of (oriented) graphs is \textit{complete} if given any two (oriented) graphs belonging to it, say $G_1, G_2 \in \mathcal{F}$, there exists a 
$G \in \mathcal{F}$ such that $G_1, G_2$ are subgraphs of $G$. 

\begin{lemma}
    Let $\mathcal{F}$ be a complete family of (oriented) graphs. Then there exists a minimal bound of $\mathcal{F}$ on $\chi_o(\mathcal{F})$ vertices. 
\end{lemma}\label{lem family bound on oriented chromatic number}

\begin{proof}
    We will prove this result using the method of 
    contradiction. Therefore, let us assume that 
    there does not exist any bound of 
    $\mathcal{F}$ on $\chi_o(\mathcal{F})$ vertices. 
    Let $\mathcal{S}$ be the set of all oriented graphs on $\chi_o(\mathcal{F})$ vertices. 
    As per our assumption, no element of $\mathcal{S}$ 
    is a bound of $\mathcal{F}$. Hence, for each 
    $\overrightarrow{X} \in \mathcal{S}$, there exists 
    a 
    $\overrightarrow{H}_{\overrightarrow{X}} \in 
    \mathcal{F}$ such that 
    $\overrightarrow{H}_{\overrightarrow{X}}$ does not 
    admit a homomorphism to $\overrightarrow{X}$. 
    
    As $|\mathcal{S}|$ is finite, and as $\mathcal{F}$ is complete, there exists an oriented graph $\overrightarrow{H} \in \mathcal{F}$ such that 
    $\overrightarrow{H}_{\overrightarrow{X}}$ is a subgraph of $\overrightarrow{H}$ for all 
    $\overrightarrow{X} \in \mathcal{S}$. 
    Notice that $\overrightarrow{H}$ does not admit a homomomorphism to any oriented graph on $\chi_o(\mathcal{F})$ vertices. Therefore, 
    $\chi_o(\overrightarrow{H}) > \chi_o(\mathcal{F})$, a contradiction. 
    \end{proof}

\begin{theorem}\label{thm orientedChromatic11}
    The oriented chromatic number for the family of triangle-free oriented graphs is at least $11$. Moreover, if $\chi_o(\mathcal{P}_4) = 11$, then any bound of $\mathcal{P}_4$ on $11$ vertices must be a tournament. 
\end{theorem}

\begin{proof}
    As $\mathcal{P}_4$ is a complete family of graphs, 
    due to Lemma~\ref{lem family bound on oriented chromatic number} there exists a minimal bound $\overrightarrow{T}$ (say) of $\mathcal{P}_4$ on $\chi_o(\mathcal{P}_4)$ many vertices. By Theorem~\ref{th P4 universal min deg} we know that 
    the minimum degree of $T$ must be $10$ or more, and thus $T$ must have at least $11$ vertices. Moreover, if $T$ has exactly $11$ vertices, then $T$ must be the complete graph. 

    As $T$ is the underlying graph of a minimal bound
    of $\mathcal{P}_4$  on $\chi_o(\mathcal{P}_4)$ 
    vertices, we must have $\chi_o(\mathcal{P}_4) \geq 
    11$, and if 
    $\chi_o(\mathcal{P}_4) = 11$, then any bound of 
    $\mathcal{P}_4$ on $11$ vertices must be a 
    tournament. 
\end{proof}

\subsection{Application 2: pushable chromatic number}\label{subsec app2}
The notion of pushable homomorphism and chromatic number was introduced by Klostermeyer and MacGillivray~\cite{push}, and have been studied since 
in a number of articles~\cite{BensmailDLNS23,BensmailNS17,DasS23,Sen17}.

Let $\overrightarrow{G}$ be an oriented graph. To \textit{push} a vertex $v$ of $\overrightarrow{G}$ is to reverse the direction of all the arcs  incident to $v$. 
Moreover, $\overrightarrow{G}$ is in a 
\textit{push relation} with $\overrightarrow{G'}$ if $\overrightarrow{G'}$ can be obtained by pushing 
a set of vertices of $\overrightarrow{G}$. 
The \textit{pushable chromatic number} of $\overrightarrow{G}$, denoted by $\chi_p(\overrightarrow{G})$, is the minimum $\chi_o(\overrightarrow{G}')$ where $\overrightarrow{G}'$ varies over all oriented graphs that are in push relation with $\overrightarrow{G}$. 

For a simple graph $G$ we have 
$$\chi_p(G) = \max\{\chi_p(\overrightarrow{G}): \overrightarrow{G} \text{ is an orientation of } G\},$$  
and for a  family $\mathcal{F}$ of  graphs we have  $$\chi_p(\mathcal{F})=\max\{\chi_p(G) : G \in \mathcal{F}\}.$$
Even though the bound $\chi_p(\mathcal{P}_4) \geq 6$ is 
not proved explicitly anywhere, it  follows directly 
from $\chi_o(\mathcal{P}_4) \geq 11$, the result 
obtained by Ochem~\cite{Ochemgirth4} via a computer 
check. As we have reproved the same theoretically, the 
bound $\chi_p(\mathcal{P}_4) \geq 6$ 
also have a theoretical proof now.

   \begin{theorem}
    The pushable chromatic number for the family of 
    triangle-free oriented graphs is at least $6$. 
\end{theorem} 

\begin{proof}
    We know $\chi_p(\mathcal{F}) \geq \lceil \frac{\chi_o(\mathcal{F})}{2} \rceil$ due to Klostermeyer and MacGillivray~\cite{push}. Therefore, 
    $$\chi_p(\mathcal{P}_4) \geq \lceil \frac{\chi_o(\mathcal{P}_4)}{2} \rceil \geq \lceil \frac{11}{2} \rceil = 6.$$
    Thus we are done.
\end{proof}

\subsection{Application 3: $2$-dipath and oriented $L(p,q)$ span}\label{subsec app3}
The Channel Assignment Problem (CAP) in wireless network is an important domain of research. There are several graph models proposed for studying CAP using graph theory. 
One of the popular such graph model of the problem is given by the $L(p,q)$-labeling problem~\cite{Calamoneri11}. While the $L(p,q)$-labeling problem is defined on simple graphs, as its less popular oriented analogue, two models exist in literature and has been studied. Those two models are called the $2$-dipath 
$L(p,q)$-labeling and the oriented $L(p,q)$-labeling~\cite{CalamoneriS13,GonccalvesRS07}.

A \textit{$2$-dipath $k$-$L(p,q)$-labeling}~\cite{CalamoneriS13} of an oriented graph $\overrightarrow{G}$ is a function 
$l : V(\overrightarrow{G}) \rightarrow \{0,1,\ldots,k\}$ 
satisfying
\begin{enumerate}[(i)]
    \item If $u,v$ are adjacent, then $|l(u) - l(v)| \geq p$.

    \item If $u,v$ are endpoints of a directed $2$-path, then 
    $|l(u) - l(v)| \geq q$.
\end{enumerate}
The \textit{$2$-dipath $L(p,q)$-labeling span} of $\overrightarrow{G}$, denoted by $\overrightarrow{\lambda}_{p,q}(\overrightarrow{G})$, 
is the minimum $k$ such that $\overrightarrow{G}$ admits a $2$-dipath $k$-$L(p,q)$-labeling.
For the family of oriented triangle-free planar graphs, we can prove an improved lower bound of the $2$-dipath $L(p,1)$-labeling span.

\begin{theorem}
    There exists an oriented triangle-free planar graph $\overrightarrow{H}$ having $$\overrightarrow{\lambda}_{p,1}(\overrightarrow{H}) \geq p+8$$ for all $p \geq 1$.  
\end{theorem}

\begin{proof}
    Let us consider the oriented graph 
    $\overrightarrow{G}_0$ 
    depicted in Fig.~\ref{fig gadget} and take $\overrightarrow{H} = \overrightarrow{G}_0$. 
    Let $l$  be any $2$-dipath $k$-$L(p,1)$-labeling of $\overrightarrow{H}$, for some non-negative integer $k$.  It is enough to show that $k \geq p+8$. 

   Notice that, the vertices from the set 
   $$R = \{x,y,v_1, v_2, \cdots, v_8\}$$ 
   must receive distinct images under $l$ as they are pairwise either adjacent or connected by a directed $2$-path. 
   Moreover, we must have $|l(z)-l(w)| \geq p$ for any $z \in \{x,y\}$ and for any $w \in \{v_1, v_2, \cdots, v_8\}$ as they are adjacent. 

    Thus, the sets 
    $$A_x = \{j : j \in \{0,1,\cdots, k\}  \text{ and } |l(x)-j| < p\}$$
    and 
    $$A_y = \{j : j \in \{0,1,\cdots, k\}  \text{ and } |l(y)-j| < p\}$$
    are sets of labels which cannot be used on any vertices of $R \setminus \{x,y\}$. As 
    $|\{l(x),l(y)\} \cup (A_x \cap A_y)| \geq p+1$ (the minimum is attained  when $l(x)=0$ and $l(y)=1$ and $l(v_i) = p+i$, for example), we must have 
    $$|l(R \setminus \{x,y\})| + |\{l(x),l(y)\} \cup (A_x \cap A_y)| \geq 8 + (p+1).$$
    However, as we also use $0$ among the labels, we have 
    $k \geq p+8$. 
\end{proof}

    The notion of oriented $L(p,q)$-labeling is a result of combining the concepts of $2$-dipath $L(p,q)$-labeling and oriented coloring~\cite{GonccalvesRS07}. However, in~\cite{0001NS18} the definition of 
oriented $L(p,1)$-labeling has an equivalent formulation 
using graph homomorphisms. As our result only concerns  oriented $L(p,1)$-labeling and its corresponding span, we are going to provide the homomorphism version of the definition. 

The graph $L_{p,k}$ is a graph on the set of vertices 
$\{0, 1, \ldots, k\}$, and its two vertices $i, j$ are adjacent if $|i-j|  \geq p$. 
An \textit{oriented $k$-$L(p,1)$-labeling}~\cite{0001NS18} of an oriented graph $\overrightarrow{G}$ is a homomorphism 
$h : \overrightarrow{G} \to  \overrightarrow{L}_{p,k}$, where $\overrightarrow{L}_{p,k}$ is any orientation of 
$L_{p,k}$. One can note that, apart from being a homomorphism, $h$ is also a $2$-dipath 
$k$-$L(p,1)$-labeling as well.

The \textit{oriented $L(p,1)$-labeling span} of $\overrightarrow{G}$, denoted by 
$\lambda^o_{p,q}(\overrightarrow{G})$, 
is the minimum $k$ such that $\overrightarrow{G}$ admits an oriented $k$-$L(p,1)$-labeling.
For the family of oriented triangle-free planar graphs, we can prove an improved lower bound of the oriented $L(p,1)$-labeling span.

\begin{theorem}
    There exists an oriented triangle-free planar graph $\overrightarrow{H}$ having 
    $$\lambda^o_{p,1}(\overrightarrow{H}) \geq 2p+8$$
    for all $p \geq 1$. 
\end{theorem}

\begin{proof}
We know that any minimal bound of $\mathcal{P}_4$ has minimum degree at least $10$ due to Theorem~\ref{th P4 universal min deg}. As any bound contains a minimal bound as a subgraph, 
any bound of $\mathcal{P}_4$ must have 
minimum degree $10$ or more as well. 

Let $\overrightarrow{H}$ be an oriented triangle-free planar graph for which the value of 
$\lambda^o_{p,1}(\overrightarrow{H})$ 
attains its maximum possible value. Say that value is $k$. 
We have to show that $k \geq 2p+8$.

Note that, in this case, the oriented graph obtained from taking the disjoint union of all the orientations of the graph $L_{p,k}$ is a bound of $\mathcal{P}_4$. However, this oriented graph (the one obtained by taking the disjoint union) has minimum degree equal to $$(k-2(p-1)) = k-2p+2.$$ However, as the minimum degree must be at least $10$ due to Theorem~\ref{th P4 universal min deg}, we have
$$k-2p+2 \geq 10 \implies k \geq 2p+8.$$
Hence we are done. 
\end{proof}

\section{Concluding remarks}\label{sec conclusion}
\begin{enumerate}
    \item If we look at Table~\ref{table results}, which summarizes the best-known bounds of the oriented relative and absolute clique numbers of planar graphs having girth at least $g \geq 3$, we observe that except for $\omega_{ro}(\mathcal{P}_3)$, that is, the oriented relative clique number of the family of planar graphs,  exact values are known for all cases. Thus, one natural direction of research is to figure out the exact value of $\omega_{ro}(\mathcal{P}_3)$. 
    This will also answer Problem~\ref{problem planar} due to Sopena~\cite{sopena-updated-survey}.

 \item One interesting observation is that the three parameters, namely, the oriented chromatic number and the oriented clique numbers take three distinct values for the family of triangle-free planar graphs. To be precise, we know that $\omega_{ao}(\mathcal{P}_4)=6$~\cite{NSS}, 
 $\omega_{ro}(\mathcal{P}_4) = 10$, and 
 $\chi_{o}(\mathcal{P}_4) \geq 11$. Thus, one of the most natural questions to ponder in this direction is to figure out the exact value of $\chi_{o}(\mathcal{P}_4)$. 
 However, it is known to be a difficult problem. Therefore, as an intermediate step towards its resolution, can we improve its existing lower and upper bounds of $11$ and $40$? 
 In particular, can we extend the ideas exercised in the proof of Theorem~\ref{th P4 universal min deg} to characterize all possible minimal bounds of $\mathcal{P}_4$ on $11$ vertices? Or possibly, improve the lower bound to $12$ by showing that there does not exist any such minimal bond of $\mathcal{P}_4$ on $11$ vertices? Furthermore, improving the upper bound of $40$ is yet another challenging question. 

 \item Finding the tight upper bound for the pushable chromatic number, $2$-dipath $L(p,1)$-span, and the oriented $L(p,1)$-span of an oriented triangle-free planar graph are natural open questions. In fact, improving the lower and upper bounds are challenging questions as well. 
\end{enumerate}

\bigskip
\noindent\textbf{Acknowledgements:}
{  This work is partially supported by the IFCAM project “Applications of graph homomorphisms” (MA/IFCAM/18/39), “NBHM/RP-8(2020)/Fresh”, SERB-MATRICS ``Oriented chromatic and clique number of planar graphs'' (MTR/2021/000858).}

\bibliographystyle{abbrv}
\bibliography{LNSS}

\end{document}